\documentclass[10pt]{article}
\usepackage{amsmath,amsthm,amsfonts,amssymb,latexsym,graphicx}

\usepackage[shortcuts]{extdash}

\usepackage{algorithm}
\usepackage[noend]{algpseudocode}
\algrenewcommand\algorithmicthen{\relax}
\algdef{C}[IF]{IF}{ElsIf}[1]{\textbf{elif}\ #1}
\algrenewcommand\algorithmicdo{\relax}

\usepackage[colorlinks=true,citecolor=blue,pdfpagemode=UseNone,pdfstartview=FitH]{hyperref}

\allowdisplaybreaks[4]

\renewcommand{\d}{\,\mathrm{d}}

\newcommand{\eP}{\lozenge}
\newcommand{\eN}{\square}

\newcommand{\E}{\mathbb{E}}
\renewcommand{\P}{\mathbb{P}}
\newcommand{\R}{\mathbb{R}}

\newcommand{\AAA}{\mathcal{A}}

\newcommand{\QQQ}{\mathcal{Q}}

\DeclareMathOperator{\cov}{cov}
\DeclareMathOperator{\var}{var}
\DeclareMathOperator{\rvar}{rvar}

\renewcommand{\complement}{\textsf{c}}
\newcommand{\DM}{\textnormal{DM}}

\theoremstyle{plain}
\newtheorem{theorem}{Theorem}[section]
\newtheorem{corollary}[theorem]{Corollary}
\newtheorem{lemma}[theorem]{Lemma}
\newtheorem{proposition}[theorem]{Proposition}

\theoremstyle{definition}

\theoremstyle{remark}
\newtheorem{remark}[theorem]{Remark}

\title{True and false discoveries with independent and sequential e-values}

\author{Vladimir Vovk\thanks%
  {Department of Computer Science,
  Royal Holloway, University of London,
  Egham, Surrey, UK.
  E-mail: \href{mailto:v.vovk@rhul.ac.uk}{v.vovk@rhul.ac.uk}.}
\and Ruodu Wang\thanks%
  {Department of Statistics and Actuarial Science,
  University of Waterloo,
  Waterloo, Ontario, Canada.
  E-mail: \href{mailto:wang@uwaterloo.ca}{wang@uwaterloo.ca}.}}

\begin{document}
\maketitle

\begin{abstract}
  In this paper we use e-values in the context of multiple hypothesis testing
  assuming that the base tests produce independent, or sequential, e-values.
  Our simulation and empirical studies and theoretical considerations suggest
  that, under this assumption, our new algorithms are superior
  to the known algorithms using independent p-values
  and to our recent algorithms designed for e-values
  without the assumption of independence.

   The version of this paper at \url{http://alrw.net/e} (Working Paper 4)
   is updated most often.
\end{abstract}

\section{Introduction}

The notion of a p-value has been widely criticized recently
as the basis of statistical hypothesis testing
(see, e.g., \cite{Wasserstein/etal:2019}).
Bayes factors provide an alternative approach
(see, e.g., \cite[Sect.~2]{Benjamin/Berger:2019}).
In this paper we continue the study of e-values, a promising new tool for testing, in addition to p-values and Bayes factors.
For recent work on hypothesis testing with e-values, as well as their connections and differences to methods based on p-values,
see, e.g., \cite{Vovk/Wang:2021,Shafer:2021,Wang/Ramdas:2022,Grunwald/etal:arXiv1906}.

Two important mathematical advantages of e-values over p-values
are that the arithmetic average of several e-values is always an e-value
and that the product of several independent e-values is an e-value.
The second property generalizes to the case of sequential e-values,
which is a far-reaching generalisation of independent e-values
to be discussed later in the paper.
Both properties are useful in multiple hypothesis testing.

Our previous papers \cite{Vovk/Wang:2021,Vovk/Wang:2023}
applying e-values to multiple hypothesis testing
did not make any assumptions about the base e-values
and relied on the arithmetic average of e-values being an e-value.
Using arithmetic averaging is very natural in this case
since the arithmetic mean essentially dominates any symmetric function
for merging e-values \cite[Proposition 3.1]{Vovk/Wang:2021}.

In this paper we assume that the base e-values are independent
or at least sequential.
This gives us extra freedom in combining e-values,
which we use for constructing confidence bounds for the number of true discoveries
in multiple hypothesis testing.
These confidence bounds are summarised in the form of ``discovery matrices'',
which were introduced in \cite{Vovk/Wang:2023}
in the case of arbitrary dependence between the base e-values.
In this paper we design similar procedures
for sequential e-values
and demonstrate the power of our procedures
in simulation and empirical studies.
The assumption that the base e-values are sequential (or independent) is,
of course,
a downside of the procedures developed in this paper;
however, under those assumptions,
the gain in power is very significant.

We start the main part of the paper in Sect.~\ref{sec:confidence}
from basic definitions related to e-values and confidence bounds.
Multiple hypothesis testing is usually understood as testing multiple hypotheses,
and this is our main object of interest in this paper.
However, Sect.~\ref{sec:single} deals with repeated testing of a single hypothesis
and introduces functions for merging e-values,
which is a first step towards testing multiple hypotheses.
In Sect.~\ref{sec:discovery} we introduce discovery matrices.
Section~\ref{sec:simulation} is devoted to a simulation study,
and Sect.~\ref{sec:empirical} to an empirical study of discovery matrices.
In Sect.~\ref{sec:efficient}
we develop a computationally more efficient version of our main procedure
for computing discovery matrices.
Statistical properties of this procedure are briefly discussed in Sect.~\ref{sec:attempt};
this is just a first attempt of theoretical analysis.
Section~\ref{sec:conclusion} concludes and lists some directions of further research.

\section{E-values for hypotheses testing and confidence estimation}
\label{sec:confidence}

To make our exposition self-contained, we start from basic definitions.
For further information on e-values,
see our previous papers \cite{Vovk/Wang:2021,Vovk/Wang:2023}.

In this paper we fix an underlying probability space $(\Omega,\AAA,\P)$
(which we rarely mention explicitly).
We will use the notation $\E(f)$ for the expectation $\int f \d\P$
of an extended random variable $f:\Omega\to[-\infty,\infty]$
w.r.\ to the true data-generating distribution $\P$.
More generally, we will write $\E_Q(f)$ for the expectation $\int f \d Q$ of $f$
w.r.\ to an arbitrary probability measure $Q$ on $(\Omega,\AAA)$.
The expectation $\E_Q(f)$ always exists (but may be equal to $\infty$)
when $f$ is nonnegative, that is $f:\Omega\to[0,\infty]$.

An \emph{e-variable} w.r.\ to a probability measure $Q$ on $(\Omega,\AAA)$
is a nonnegative extended random variable $E:\Omega\to[0,\infty]$
such that $\E_Q(E)\le1$.
A large value of $E(\omega)$ for the realized outcome $\omega$
is interpreted as casting doubt on $Q$ being the true data-generating distribution $\P$.
Indeed, by Markov's inequality,
$Q(E\ge c)\le1/c$ for any $c>1$,
and so $E$ can take large values only with a small $Q$-probability.
This interpretation assumes, of course,
that $E$ was chosen in advance of the statistical experiment.
An \emph{e-value} is the value $E(\omega)$ taken by the e-variable.

A more standard way of testing statistical hypotheses
is to use \emph{p-values},
defined to be the values taken by \emph{p-variables},
i.e., nonnegative random variables $P$ such that,
for any $\alpha\in(0,1)$, $Q(P\le\alpha)\le\alpha$.
There are numerous ways of converting p-values to e-values
and a natural way of converting e-values to p-values
(see, e.g., \cite[Sect.~2]{Vovk/Wang:2021}),
but in order to compare the strength of evidence against the null hypothesis
provided by e-values and p-values we will use Jeffreys's rules of thumb
that we will describe in Sect.~\ref{sec:simulation}.

Alongside the true data-generating probability measure $\P$
on the underlying measurable space $(\Omega,\AAA)$
we will consider other probability measures
and will use the notation $\QQQ$
for the set of all probability measures on $(\Omega,\AAA)$.
An \emph{e-test} is a family of e-variables $(E_Q\mid Q\in\QQQ)$,
$E_Q$ being an e-variable w.r.\ to $Q$.
Assuming the e-test $E$ had been chosen before the experiment,
the interpretation of $E_Q(\omega)$ is that
it measures lack of agreement between an outcome $\omega$
and a putative explanation $Q$.
Namely, suppose $E_Q(\omega)$ is large;
then the disagreement between $Q$ and $\omega$ has two sides:
\begin{itemize}
\item
  if $\omega$ happens,
  we do not regard $Q$ to be feasible as a data-generating distribution;
\item
  if we believe that $Q$ is the data-generating distribution ($Q=\P$),
  we do not expect $\omega$ to happen.
\end{itemize}
If $E_Q(\omega)\ge\alpha$,
we will say that $\omega$ is \emph{$\alpha$-strange} under $Q$.
Otherwise, $\omega$ and $Q$ \emph{agree} at level $\alpha$.

\subsection{Confidence regions}

We are often interested not in the data-generating distribution $\P$ itself
but in the value $g(\P)$ of a function $g$ on it;
for example $g$ can map the probability measures on the real line $\R$
to their medians.
In confidence estimation,
$g(\P)$ is usually interpreted as the value of a parameter
corresponding to $\P$;
we will refer to $g$ as a \emph{parameter function}.
For a given e-test $E$, parameter function $g:\QQQ\to\Theta$,
\emph{significance level} $\alpha>0$,
and realized outcome $\omega\in\Omega$,
the corresponding \emph{confidence region} is defined as
\begin{equation}\label{eq:region}
  \Gamma^g_{\alpha}(\omega)
  :=
  \{g(Q)\mid E_Q(\omega)<\alpha\}.
\end{equation}
A typical value of $\alpha$ used in this paper is $10$;
in this case $\Gamma^g_{\alpha}\subseteq\Theta$
($\Theta$ being the chosen \emph{parameter space})
is the set of all parameter values that agree with the realized outcome
at level $10$.

Notice that the confidence regions \eqref{eq:region}
are valid simultaneously for all functions $g:\QQQ\to\Theta$
(let us assume that the parameter space $\Theta$ is fixed,
although our statement of simultaneous validity
is also true for variable $\Theta$),
provided we are using the same e-test for all $g$.
Namely, $g(Q)\in\Gamma^g_{\alpha}(\omega)$
for all $g$ unless $\omega$ is $\alpha$-strange under $Q$.

\subsection{Necessity and possibility measures}

These are different ways of packaging the same information
as that available in confidence regions \eqref{eq:region}.
Given an e-test $E$,
we define the \emph{necessity measure}
\[
  \eN^g(B\mid\omega)
  :=
  \inf_{Q:g(Q)\notin B}
  E_Q(\omega)
\]
and the \emph{possibility measure}
\begin{equation}\label{eq:eP}
  \eP^g(B\mid\omega)
  :=
  \inf_{Q:g(Q)\in B}
  E_Q(\omega)
  =
  \eN^g(B^{\complement}\mid\omega),
\end{equation}
where $B\subseteq\Theta$ is a set of parameter values
and $B^{\complement}:=\Theta\setminus B$ is its complement.

The intuition behind $\eN^g(B\mid\omega)$
is given in terms of a Fisher-type disjunction:
$g(\P)\in B$ unless $\omega$ is $\eN^g(B\mid\omega)$-strange under $\P$
(cf.\ \cite[Sect.~III.1]{Fisher:1973}).
For example, we expect $g(\P)\in B$ if $\eN^g(B\mid\omega)$ is large.
This also gives an interpretation of $\eP^g(B\mid\omega)$
in view of the duality
\[
  \eP^g(B\mid\omega)
  =
  \eN^g(B^{\complement}\mid\omega).
\]
Namely, $g(\P)\in B$ is impossible
unless $\omega$ is $\eP^g(B\mid\omega)$-strange.
(So we could also call \eqref{eq:eP} an impossibility measure.)
For example, we expect $g(\P)\notin B$ if $\eP^g(B\mid\omega)$ is large.

\section{Multiple testing of a single null hypothesis}
\label{sec:single}

A measurable function $F:[0,\infty)^K\to[0,\infty)$
for an integer $K\ge1$ is an \emph{ie-merging function}
if, for any probability space and any independent e-variables $E_1,\dots,E_K$ on it,
the extended random variable $F(E_1,\dots,E_K)$ is an e-variable
(when it exists, i.e., when $E_1,\dots,E_K$ take finite values).

An important relaxation of independence
of e-variables $E_1,\dots,E_K$ is the requirement that they be \emph{sequential}:
for each $k\in\{1,\dots,K\}$,
\[
  \E(E_k\mid E_1,\dots,E_{k-1})\le1.
\]
Sequential e-variables typically describe a sequential process of testing
the null hypothesis $\P$:
$E_k$ should be a valid e-variable even when $E_1,\dots,E_{k-1}$ are known.
A measurable function $F:[0,\infty)^K\to[0,\infty)$ is an \emph{se-merging function}
if, for any probability space and any sequential e-variables $E_1,\dots,E_K$ on it,
$F(E_1,\dots,E_K)$ is an e-variable.
Of course, all se-merging functions are ie-merging functions,
but the converse is not true \cite[Example~2]{Vovk/Wang:2024}.

We extend these merging functions to be of the form $F:[0,\infty]^K\to[0,\infty]$
in a canonical way (described in \cite{Vovk/Wang:2021}, Proposition C.1 in the Online Supplement):
if any of the arguments of $F$ is $\infty$, set $F$ to~$\infty$.

Important examples of se-merging functions \cite{Vovk/Wang:2021} are
\begin{equation}\label{eq:U}
  U_n(e_1,\dots,e_K)
  :=
  \frac{1}{\binom{K}{n}}
  \sum_{\{k_1,\dots,k_n\}\subseteq\{1,\dots,K\}}
  e_{k_1} \dots e_{k_n},
  \quad
  n\in\{1,\dots,K\}.
\end{equation}
We will refer to them as the \emph{U-statistics}
(they are the standard U-statistics with product as kernel).
The statistics $U_1$ play a special role
since they belong to the narrower class of \emph{e-merging functions},
meaning that $U_1(E_1,\dots,E_k)$ is an e-variable
whenever $E_1,\dots,E_K$ are e-variables
(not necessarily sequential).
Multiple hypothesis testing using $U_1$ was explored in \cite{Vovk/Wang:2021,Vovk/Wang:2023},
and in this paper we will often be interested in $U_2$.

An assumption that we will need to make about merging functions in this paper
is that they are increasing in each argument
and are symmetric (more fully, permutation-symmetric, i.e., not depending on the order of their arguments).
In particular, Algorithm~\ref{alg:DM} below
will assume that its underlying merging function $F$
is increasing and symmetric.
While this assumption appears to be very natural for ie-merging functions,
it becomes more restrictive for se-merging functions:
see Example~1 in~\cite{Vovk/Wang:2024}.

\section{Discovery matrices for independent and sequential e-values}
\label{sec:discovery}

In this section we consider the problem of testing $K$ statistical hypotheses
$H_k$, $k=1,\dots,K$, for some $K\in\{2,3,\dots\}$.
Each hypothesis $H_k$ may be \emph{simple}, $H_k\in\QQQ$,
or \emph{composite}, $H_k\subseteq\QQQ$.
We assume, without loss of generality,
that $H_k$ is composite
(a simple hypotheses $Q\in\QQQ$ can be interpreted as the composite hypothesis $\{Q\}$).

For each $k$ we are testing the null hypothesis $H_k$ using an e-variable $E_k$.
Namely, we are given a sequence $E_1,\dots,E_K$ of extended random variables
such that each $E_k$, $k=1,\dots,K$,
is an e-variable for testing $H_k$
in the sense of satisfying $\E_Q(E_k) \le 1$ for all $Q\in H_k$.
Therefore, $E_k$ is an e-variable w.r.\ to any $Q\in H_k$,
but it does not have to be an e-variable w.r.\ to the true data-generating measure $\P$.

\subsection{Discovery vectors}

We are interested in the set of $k\in\{1,\dots,K\}$ for which $\P\in H_k$,
i.e., $H_k$ is a true hypothesis (which  also determines the set of indices for which the hypothesis is false).
Our first goal is to define, for each set $R\subseteq\{1,\dots,K\}$,
a confidence region for the number of $k\in R$
for which $H_k$ is false.
Our confidence region will always be of the form $\{L,\dots,K\}$
for some lower bound $L$
(we only have a lower confidence bound
since $\E(E_k)\le1$ may be true even if $H_k$ is not a true hypothesis).

Suppose that $H_1,\dots,H_K$ are scientific hypotheses
that an experimental scientist is interested in.
After observing the e-values $e_1,\dots,e_K$,
where $e_k$ is the realized value of the e-variable $E_k$,
the scientist comes up with a \emph{rejection set} $R\subseteq\{1,\dots,K\}$
containing the indices of the hypotheses that she decides to reject
based on $e_1,\dots,e_K$.
The elements of $R$ are referred to as \emph{discoveries}
(i.e., they are discoveries as claimed by the scientist).
A discovery $k\in R$ is a \emph{true discovery} if $\P\notin H_k$,
and it is a \emph{false discovery} if $\P\in H_k$.
The motivation for this terminology is that $H_k$ is regarded
to be the default hypothesis saying that an interesting effect does not exist
\cite{Kotz/etal:2006};
rejecting the null hypothesis is then a discovery.

A natural way to choose the rejection set $R$
is to include in it a number of $k$ with the largest $e_k$,
but the scientist sometimes might want to include all $k$
connected by some common theme \cite[Sect.~4.1]{Goeman/Solari:2011}.
We will start from discussing arbitrary $R$
and then will move on to the $R$ corresponding to the largest e-values.

For each $Q\in\QQQ$, let
\[
  I_Q
  :=
  \left\{
    k\in\{1,\dots,K\}
    \mid
    Q\in H_k
  \right\}
\]
be the indices of the true hypotheses under $Q$.
Given a rejection set $R$, we are interested in the parameter
\begin{equation*}
  g_R(Q)
  :=
  \left|
    R \setminus I_Q
  \right|,
\end{equation*}
which is the number of true discoveries in $R$
under $Q$.
While in this paper we concentrate on the parameter function $g_R$,
this function can be generalized in various directions;
see, e.g., \cite[Remark~6.1]{Vovk/Wang:2023}.

We are interested in lower confidence bounds
on the number of true discoveries.
For that,
in order to apply the recipe \eqref{eq:region},
we need an e-test.
A natural way to obtain an e-test is to apply
a sequence of symmetric e-merging functions $F_{K'}:[0,\infty)^{K'}\to[0,\infty)$,
$K'\in\{2,\dots,K\}$,
at each $Q\in\QQQ$:
\begin{equation}\label{eq:test}
  E_Q
  :=
  F_{\left|I_Q\right|}((E_k)_{k\in I_Q})
  =
  F((E_k)_{k\in I_Q}),
\end{equation}
where $(E_k)_{k\in I_Q}$ stands for the sequence of length $I_Q$
whose elements are $E_k$ with $k\in I_Q$
in the ascending order of $k$
(although the order of the elements of this sequence does not matter
since $F_{K'}$ are assumed symmetric).
We define $F_0$ and $F_1$ in~\eqref{eq:test}
(needed when $\left|I_Q\right|\in\{0,1\}$)
separately;
namely, $F_0:=1$ and $F_1(e):=e$.
The second ``$=$'' in \eqref{eq:test} introduces our convention of omitting
the lower index $K'$ in the notation $F_{K'}$;
indeed, the lower index is redundant as it can be recovered
from the number of arguments of $F_{K'}=F_{K'}(e_1,\dots,e_{K'})$.

In order to obtain tighter confidence regions,
we assume that, for any probability measure $Q\in\QQQ$,
the random e-variables $E_k$, $k\in I_Q$, are independent under $Q$.
It is clear that in this case \eqref{eq:test} is an e-test
provided $F$ is a symmetric ie-merging function.
This will be the main case considered in this paper.

We gave examples of useful ie-merging functions in Sect.~\ref{sec:single}.
However, for the use in~\eqref{eq:test}
we need to generalize \eqref{eq:U} to the case $n>K\ge1$
(remember that we have $K':=\left|I_Q\right|$ in place of $K$
when using \eqref{eq:U} in the context of \eqref{eq:test},
and so $K'$ can be a small number even for a large number of the null hypotheses).
Let us set
\begin{equation}\label{eq:n-K}
  U_n(e_1,\dots,e_K)
  :=
  U_K(e_1,\dots,e_K)
  =
  e_1\dots e_K,
  \quad
  n>K.
\end{equation}
This convention is not needed in the case $n=2$,
which is our main object of study,
since it was already made implicitly when defining $F_1$.

For each $j\in\{0,\dots,K-1\}$,
we are interested in the possibility measure
\[
  \eP^{g_R}(\{0,\dots,j\}\mid\omega),
\]
where $\omega\in\Omega$ is the realized outcome,
which we usually omit, as is customary in probability theory.
The interpretation of $\eP^{g_R}(\{0,\dots,j\})$
is that the number of true discoveries exceeds $j$
unless the realized outcome is $\eP^{g_R}(\{0,\dots,j\})$-strange.

A more explicit representation of $\eP^{g_R}(\{0,\dots,j\})$ is:
\begin{equation}\label{eq:explicit}
  \begin{aligned}
    \eP^{g_R}(\{0,\dots,j\})
    &=
    \inf_{Q\in\QQQ:\left|R \setminus I_Q\right|\le j}
    E_Q
    =
    \min_{Q\in\QQQ:\left|R \setminus I_Q\right|\le j}
    F
    ((E_k)_{k\in I_Q})\\
    &\ge
    \min_{I\subseteq\{1,\dots,K\}:\left|R \setminus I\right|\le j}
    F
    ((E_i)_{i\in I})
    =:
    D^R_F(j).
  \end{aligned}
\end{equation}
The inequality in \eqref{eq:explicit} may be strict
(when some subsets of $\{1,\dots,K\}$
cannot be represented in the form $I_Q$ for any $Q\in\QQQ$),
but in interesting cases we have an equality there.
The symbol ``$=:$'' means that $D^R_F(j)$ is being defined.

The vector $(D^R_F(j))_{j=0}^{K-1}$ is the \emph{discovery vector} for $R$
(in \cite{Vovk/Wang:2021} we referred to it as regularized discovery e-vector).
We will now introduce ``discovery matrices'',
whose rows are discovery vectors.

\subsection{Discovery matrices}

In this subsection we impose restrictions
on the ie-merging function $F$ in \eqref{eq:explicit}
which will allow us to restrict our attention
to a relatively small number of rejection sets.
Namely, as mentioned at the end of Sect.~\ref{sec:single},
we assume that $F$ is increasing in each of its arguments and symmetric.
Suppose the e-values are listed in the decreasing order,
$e_1\ge\dots\ge e_K$.
Then the rejection sets
\[
  R_r
  :=
  \{1,\dots,r\},
  \quad
  r\in\{1,\dots,K\},
\]
form a complete family of rejection sets:
for any $r$ and any rejection set $R$ of size $r$
we have, for all $j\in\{0,\dots,r-1\}$,
$D^{R_r}_F(j)\ge D^{R}_F(j)$.

\begin{algorithm}[bt]
  \caption{Discovery matrix (lower triangular)}
  \label{alg:DM}
  \begin{algorithmic}[1]
    \Require
      ie-merging (or se-merging) functions $F_k$, $k\in\{1,\dots,K\}$.
    \Require
      a decreasing sequence of e-values $e_1\ge\dots\ge e_K$.
    \For{$r=1,\dots,K$}
      \For{$j=0,\dots,r-1$}\label{l:middle}
        \State $S_{r,j}:=\{j+1,\dots,r\}$
        \State $\DM^F_{r,j}:=F((e_i)_{i\in S_{r,j}})$
        \For{$k=r+1,\dots,K$}\label{l:inner}
          \State $e := F((e_i)_{i\in S_{r,j}\cup\{k,\dots,K\}})$
          \If{$e < \DM^F_{r,j}$}
	    \State $\DM^F_{r,j} := e$\label{l:e}
          \EndIf
        \EndFor
      \EndFor
    \EndFor
  \end{algorithmic}
\end{algorithm}

The \emph{discovery matrix} is a lower triangular matrix with the entries
\[
  \DM_{r,j}
  =
  \DM^F_{r,j}
  :=
  D^{R_r}_F(j),
  \quad
  r=1,\dots,K,
  \enspace
  j=0,\dots,r-1;
\]
we often drop the upper index $F$ in $\DM^F_{r,j}$ when it is clear from the context.
Algorithm~\ref{alg:DM} is one way of constructing a discovery matrix,
under an additional assumption
(cf.\ \eqref{eq:condition} below),
based on a family of ie-merging functions $F_k$, $k\in\{2,\dots,K\}$.
The algorithm implements the equality between the extreme terms of
\begin{equation}\label{eq:algorithm}
  \DM^F_{r,j}
  =
  \min_{I\subseteq\{1,\dots,K\}:\left|R^r\setminus I\right|\le j}
  F((E_i)_{i\in I})
  =
  \min_{I\subseteq\{1,\dots,K\}:\left|R^r\setminus I\right|=j}
  F((E_i)_{i\in I}).
\end{equation}
The first equality in~\eqref{eq:algorithm}
follows from the last equality (``$=:$'') in \eqref{eq:explicit},
and the second equality in~\eqref{eq:algorithm} follows
from the following natural condition of consistency
between different merging functions in the family $(F_k)$:
if $\mathbf{e}\in[0,\infty)^k$ for $k\in\{1,\dots,K-1\}$ and $e\in[0,\infty)$,
then
\begin{equation}\label{eq:condition}
  e \ge \max(\mathbf{e})
  \Longrightarrow
  F(\mathbf{e},e)
  \ge
  F(\mathbf{e}).
\end{equation}
This is the condition under which Algorithm~\ref{alg:DM} is valid.

\begin{remark}
  Algorithm~\ref{alg:DM} is a version of Algorithm~2 in \cite{Vovk/Wang:2023}.
  In both algorithms the e-values are assumed to be ordered
  (without loss of generality, since we assume $F$ to be symmetric),
  but while in Algorithm~2 in \cite{Vovk/Wang:2023} the order is ascending,
  in Algorithm~\ref{alg:DM} it is descending.
  Another difference is that in this paper
  (as in the arXiv version of \cite{Vovk/Wang:2023})
  we use the version of the discovery matrix described
  in \cite[Remark~6.3]{Vovk/Wang:2023}.
  Finally, in \cite{Vovk/Wang:2023} we were mainly interested
  in the arithmetic-mean merging function
  (because it essentially dominates all symmetric e-merging functions
  with no dependence assumption),
  while in this paper we are interested in a wider range of merging functions
  by considering independent or sequential e-variables.
\end{remark}

The validity of Algorithm~\ref{alg:DM} follows
from the fact that the min in
\[
  \min_{I\subseteq\{1,\dots,K\}:\left|R^r\setminus I\right|=j}
  F((E_i)_{i\in I})
\]
(cf.\ \eqref{eq:algorithm}) is attained at $I$
of the form
\[
  S_{r,j}\cup\{k,\dots,K\}
  =
  \{j+1,\dots,r\}\cup\{k,\dots,K\}.
\]
This follows immediately from the ie-merging function $F$
being increasing in all arguments and symmetric.

Discovery matrices satisfy useful properties of monotonicity
given in the following proposition
(which is proved as part of the proof of Proposition~6.3
in the arXiv version of \cite{Vovk/Wang:2023}).
The following proposition assumes the standard way of presenting matrices
(as in Figures~\ref{fig:U1_vs_U2}--\ref{fig:GWGS} below).
\begin{proposition}
  Suppose the family $(F_k)$ of ie-merging functions satisfies \eqref{eq:condition}.
  Then $\DM_{r,j}$ is
  \begin{itemize}
  \item
    decreasing along the rows:
    $\DM_{r,j'}\le\DM_{r,j}$ when $j'>j$;
  \item
    increasing down the columns:
    $\DM_{r',j}\ge\DM_{r,j}$ when $r'>r$;
  \item
    decreasing in the Southeast direction:
    $\DM_{r+i,j+i}\le\DM_{r,j}$ when $i>0$.
  \end{itemize}
  (Even if \eqref{eq:condition} is violated,
  these properties are satisfied for the regularized version~\eqref{eq:regular}
  defined below.)
\end{proposition}

Condition~\eqref{eq:condition},
despite looking very natural,
is not satisfied, strictly speaking,
even for the $U_2$ ie-merging function.
For example, taking $\mathbf{e}=(e_1)$ of length 1, we have
\[
  U_2(e_1,e)
  =
  e_1 e
  <
  e_1
  =
  U_2(e_1)
  =
  U_1(e_1)
\]
(remember our convention \eqref{eq:n-K})
for some $e_1$ and $e$ such that $e>e_1$:
it suffices to take $e<1$.
In this case, however, both $e_1$ and $e_1 e$
are useless as e-values,
and in practice we can regard~\eqref{eq:condition} as satisfied for $U_2$
(it is natural to ignore entries below 1 in discovery matrices,
or to replace them by 1).

If \eqref{eq:condition} is violated,
we should regularize the output of Algorithm~\ref{alg:DM} by redefining
\begin{equation}\label{eq:regular}
  \DM^F_{r,j}
  :=
  \min
  \left(
    \DM^F_{r,0},\dots,\DM^F_{r,j}
  \right).
\end{equation}
Then the equality between the extreme terms of \eqref{eq:algorithm}
will always hold,
and so $\DM_{r,j}$ will be a lower bound on $\eP^{g_R}(\{0,\dots,j\})$.

The computation time of Algorithm~\ref{alg:DM}
depends on the computation time of the family of ie-merging function $F_k$, $k=1,\dots,K$.
In the cases of primary interest to us, namely $U_n$ for a fixed $n$,
the computation time for each of $F_k$ is $O(K)$ (and even $O(k)$).
For $U_1$ this is obvious,
and for $U_2$ it suffices to represent it in the form
\[
  U_2(e_1,\dots,e_k)
  =
  \frac{1}{K(K-1)}
  \left(
    (e_1+\dots+e_k)^2
    -
    (e_1^2+\dots+e_k^2)
  \right).
\]
For a general fixed $n$, use induction in $n$.

Assuming that $F_k$, $k=1,\dots,K$, are computed in linear time, $O(K)$,
each row of the discovery matrix can be computed by Algorithm~\ref{alg:DM}
in time $O(K^3)$:
computing each e-value (line \ref{l:e}) takes time $O(K)$,
each inner loop (line \ref{l:inner}) takes time $O(K^2)$,
and so each middle loop (line \ref{l:middle}) takes time $O(K^3)$.
Therefore, computing the whole discovery matrix takes time $O(K^4)$.
Applying regularization~\eqref{eq:regular} on top of these calculations,
if needed, only takes time $O(K)$,
and so does not change the overall time
($O(K^3)$ in the case of a row of the discovery matrix
or $O(K^4)$ in the case of the whole discovery matrix).

\subsection{Beyond independent e-values}

So far we have discussed the case of independent e-variables $E_k$
(more precisely, we have assumed the $Q$-independence of $(E_k\mid Q\in H_k)$
under any $Q\in\QQQ$).
Replacing ie-merging functions by se-merging functions,
we can assume, instead, that $E_1,\dots,E_K$ are sequential,
in the sense of satisfying $\E_Q(E_k\mid E_1,\dots,E_{k-1}) \le 1$ whenever $Q\in H_k$.
To check that \eqref{eq:test} is indeed an e-test,
apply the tower property of conditional expectations:
the latter implies that $(E_k)_{k\in I_Q}$ are sequential e-variables under $Q$
whenever $E_1,\dots,E_K$ are.

In Sections~\ref{sec:simulation} and~\ref{sec:empirical}
we will apply Algorithm~\ref{alg:DM} to $K$ independent e-values.
A typical scenario in which independent e-values arise
in $K$ statistical experiments
is where there are reasons to believe that there is no or little interference
or sample sharing between different experiments.
The experiments can be conducted at the same time or at different times.

\begin{remark}
  It is interesting that,
  in the case of $U_2$ considered in the following two sections,
  the assumption of independence
  between $E_1,\dots,E_K$ can be relaxed
  in a way different from this paper's assumption
  that $E_1,\dots,E_K$ are sequential.
  Instead, we can assume that the covariances $\cov(E_i,E_j)$, $i\ne j$, are all nonpositive.
  Indeed, in this case
  \begin{multline*}
    \E(E_i E_j)
    =
    \E((\E E_i+(E_i-\E E_i))(\E E_j+(E_j-\E E_j)) \\
    =
    (\E E_i) (\E E_j)
    +
    \cov(E_i,E_j)
    \le
    1.
  \end{multline*}
  Chi et al.\ \cite[Sect.~2]{Chi/etal:arXiv2212} describe interesting practical situations
  where this assumption is realistic.
  This includes sampling without replacement.
  Suppose there is a very large fixed population,
  and $K$ scientists each take a different sample from the population without replacement.
  Then any increasing statistics (e.g., sample means) of these $K$ subsamples are negatively associated,
  and their correlation is nonpositive.
\end{remark}

A possible scenario in which sequential (but not independent) e-variables $E_1,\dots,E_K$ arise
is where the experiments are performed sequentially,
the experiment $E_k$ for $k\in\{2,\dots,K\}$ is started after the experiment $E_{k-1}$ is finished,
and the experiment $E_k$ is designed using the results
of the previous experiments $E_1,\dots,E_{k-1}$.
It is essential that $E_k$ should be an e-variable
conditionally on the previous experiments.
The computational experiments reported in this paper do not cover this case,
but it is important for us mathematically
as all our methods are applicable to sequential e-variables.

\section{A simulation study with independent e-values}
\label{sec:simulation}

In this section we run Algorithm~\ref{alg:DM} applied to $U_2$ and, for comparison, $U_1$.
As discussed in the previous section,
the $U_n$ discovery matrix can be computed in time $O(K^4)$.
For $n=1$, the time can be improved from $O(K^4)$ to $O(K^2)$
\cite[Sect.~9]{Vovk/Wang:2023}.
For $n=2$, we can improve the time $O(K^4)$ to $O(K^3)$,
as we show in Sect.~\ref{sec:efficient} (Algorithm~\ref{alg:DM-eff}),
and this is sufficient to cope not only with the case $K=200$
that we use in our simulation studies in this section
but also with $K$ of a few thousand
(as used in our empirical studies in the next section,
where we compute only part of the discovery matrix).

\begin{figure}
  \begin{center}
    \includegraphics[width=0.49\textwidth]{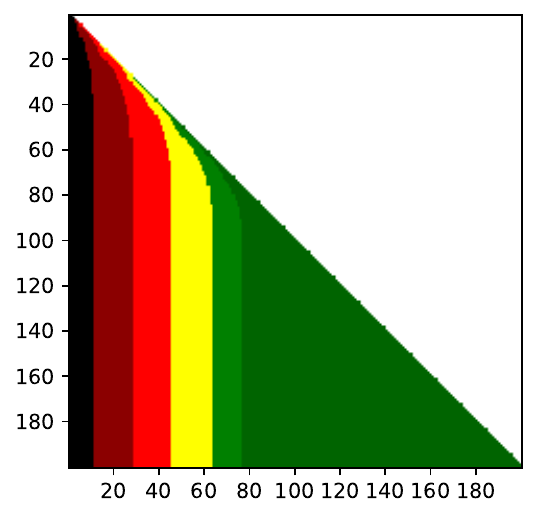}
    \includegraphics[width=0.49\textwidth]{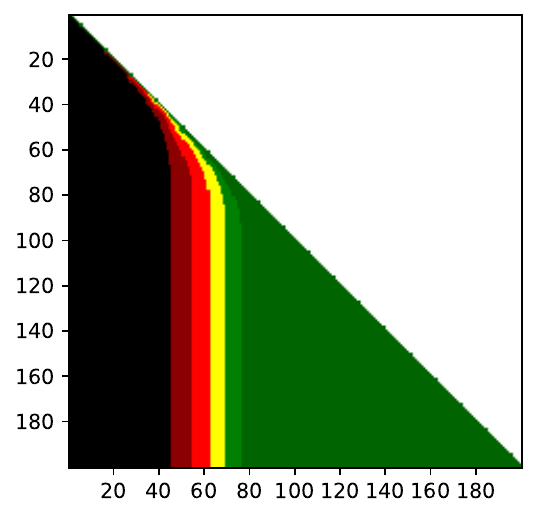}
  \end{center}
  \caption{Left panel:
    the discovery matrix for the $U_1$ statistic (i.e., arithmetic mean)
    for 100 false and 100 true null hypotheses.
    Right panel: the $U_2$ analogue.}
  \label{fig:U1_vs_U2}
\end{figure}

We generate the base e-values as in \cite[Sect.~7]{Vovk/Wang:2023}:
the null hypotheses are $\mathcal{N}(0,1)$, $K=200$,
the first $100$ observations $x$ are generated from $\mathcal{N}(-3,1)$,
the last $100$ from $\mathcal{N}(0,1)$, all independently,
and the base e-variables are the likelihood ratios
\begin{equation}\label{eq:E}
  E(x)
  :=
  \frac{\exp(-(x+3)^2/2)}{\exp(-x^2/2)}
  =
  \exp(-3x - 9/2).
\end{equation}

The results are shown in Figure~\ref{fig:U1_vs_U2}
(whose left panel is identical to the upper left panel of Figure~4 in \cite{Vovk/Wang:2023});
they are much better for $U_2$.
Our chosen colour scheme, to be described momentarily,
entails that ``better'' means ``darker'' here.

Each panel shows the lower triangular matrix $\DM^F_{r,j}$,
the left for $F=U_1$ and the right for $F=U_2$.
The colour scheme used in this figure is inspired by Jeffreys's \cite[Appendix~B]{Jeffreys:1961}
(as in \cite{Vovk/Wang:2023}):
\begin{itemize}
\item
  The entries with $\DM^F_{r,j}$ below 1 are shown in dark green;
  there is no evidence that there are more than $j$ true discoveries
  among the $r$ hypotheses with the largest e-values.
\item
  The entries $\DM^F_{r,j}\in(1,\sqrt{10})\approx(1,3.16)$ are shown in light green.
  For them the evidence is poor.
\item
  The entries $\DM_{r,j}\in(\sqrt{10},10)\approx(3.16,10)$ are shown in yellow.
  The evidence is substantial.
\item
  The entries $\DM^F_{r,j}\in(10,10^{3/2})\approx(10,31.6)$ are shown in light red.
  The evidence is strong.
\item
  The entries $\DM^F_{r,j}\in(10^{3/2},100)\approx(31.6,100)$ are shown in dark red.
  The evidence is very strong.
\item
  Finally, the entries $\DM^F_{r,j}>100$ are shown in black,
  and for them the evidence is decisive.
\end{itemize}

The interpretation of the two discovery matrices in Figure~\ref{fig:U1_vs_U2}
in terms of confidence regions
is particularly convenient:
for each row $r$ of a discovery matrix,
\begin{itemize}
\item
  the red (both dark and light), yellow, and green (both dark and light) entries in row $r$
  coincide with the confidence region at significance level $100$
  for the number of true discoveries among the $r$ largest e-values,
\item
  the light red, yellow, and green entries in row $r$
  coincide with the confidence region at significance level $10^{3/2}$
  for the number of true discoveries among the $r$ largest e-values,
\item
  the yellow and green entries in row $r$
  coincide with the confidence region at significance level $10$
  for the number of true discoveries among the $r$ largest e-values,
\item
  the green entries in row $r$
  coincide with the confidence region at significance level $\sqrt{10}$
  for the number of true discoveries among the $r$ largest e-values
  (this is the final interesting case).
\end{itemize}
A darker discovery matrix
(such as the right-hand panel of Figure~\ref{fig:U1_vs_U2}
as compared with its left-hand panel)
is preferable since it means tighter confidence regions.

\subsection*{Comparison with methods based on p-values}

It is interesting to compare the right panel of Figure~\ref{fig:U1_vs_U2}
to a similar figure for p-values obtained by standard methods.
A rigorous way of doing it
would be to convert, or as we usually say, \emph{calibrate},
e-values to p-values and vice versa.
As we mentioned in Sect.~\ref{sec:confidence},
there are numerous ways of calibrating p-values to e-values.
However, the only admissible way of calibrating an e-value $e$
to a p-value $p$ is $p:=\min(1/e,1)$.
Its validity follows from Markov's inequality
($\E(1/E\le\alpha)\le\alpha$ for any e-variable $E$),
and its domination of any other e-to-p calibrator
is stated in \cite[Proposition~2.2]{Vovk/Wang:2021}.
Comparing e-values and p-values is discussed in detail
in \cite[Sect.~3]{Vovk/Wang:2023},
where we demonstrate very low ``round-trip efficiency''
of converting e-values to p-values and back:
namely, we start from a highly statistically significant p-value of $0.5\%$,
transform it to an e-value using a popular method
\cite[(1)]{Benjamin/Berger:2019},
and then transform it back to a p-value by inverting the e-value;
the resulting p-value (7.2\%) is not even statistically significant.
Therefore, in \cite{Vovk/Wang:2021} we emphasize informal comparisons.

\begin{figure}
  \begin{center}
    \includegraphics[width=0.49\textwidth]{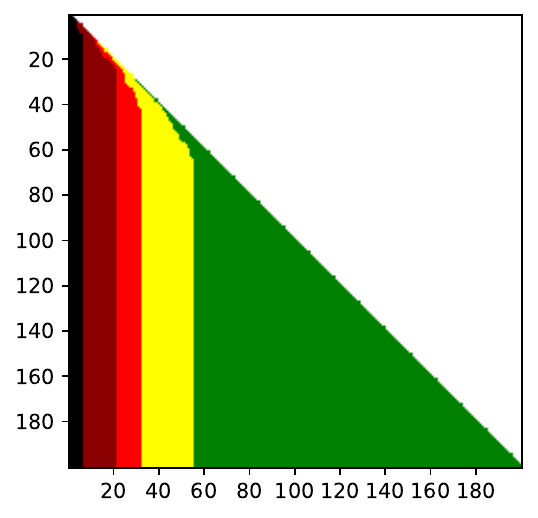}
    \includegraphics[width=0.49\textwidth]{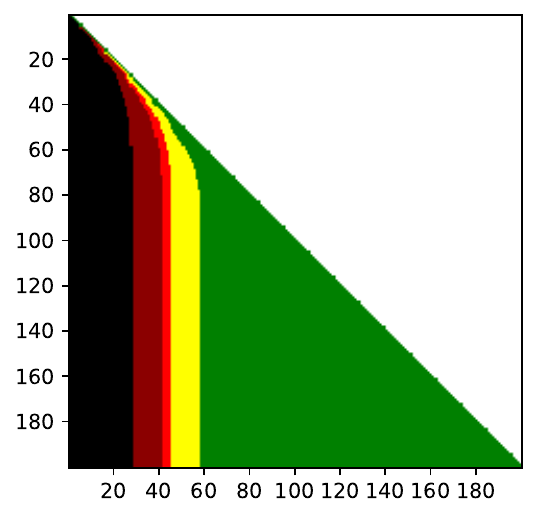}
  \end{center}
  \caption{Left panel:
    the discovery p-matrix for the GWGS procedure
    with independent p-values.
    Right panel: the $U_2$ discovery matrix e-to-p calibrated
    via $p:=\min(1/e,1)$.}
  \label{fig:ie_vs_p}
\end{figure}

The left panel of Figure~\ref{fig:ie_vs_p} shows the p-values produced
by the GWGS procedure, which is designed specifically for p-values.
The procedure is very standard
\cite{Genovese/Wasserman:2004,Goeman/Solari:2011,Goeman/etal:2019Biometrika},
but the abbreviation,
after the four authors of this procedure,
that we use for it was introduced in our previous paper \cite{Vovk/Wang:2023}.
As presented in that paper, the procedure produces a discovery matrix based on p-values,
or as we will say, \emph{discovery p-matrix};
namely, the discovery p-matrix in Figure~\ref{fig:ie_vs_p} has the entries
\begin{equation}\label{eq:DM-p}
  \DM^{\text{Simes}}_{r,j}
  :=
  \max_{I\subseteq\{1,\dots,K\}:\left|R^r\setminus I\right|\le j}
  F((P_i)_{i\in I}),
  \quad
  r=1,\dots,K,
  \enspace
  j=0,\dots,r-1
\end{equation}
(cf.\ \eqref{eq:algorithm}),
where $F$ is Simes's \cite{Simes:1986} function for merging independent p-values.
The underlying p-value $P_k$ for testing the null hypothesis $H_k$
for each of the 200 observations is the optimal one coming from the Neyman--Pearson lemma
(and so it is based on \eqref{eq:E} as test statistic).

Figure~\ref{fig:ie_vs_p} uses what we called Fisher's scale in \cite{Vovk/Wang:2023},
but now we extend it by two further thresholds,
one of which is $0.5\%$, as advocated by Benjamin et al.\ \cite{Benjamin/etal:2018}.
Therefore, our colour scheme is:
\begin{itemize}
\item
  The p-values above $5\%$ are shown in green;
  they are not statistically significant.
\item
  The p-values between $1\%$ and $5\%$ are shown in yellow;
  they are statistically significant but not highly significant.
\item
  The p-values between $0.5\%$ and $1\%$ are shown in red;
  they are highly significant
  (but fail to attain the more stringent criterion of significance
  advocated in \cite{Benjamin/etal:2018}).
\item
  The p-values between $0.1\%$ and $0.5\%$ are shown in dark red.
\item
  The p-values below $0.1\%$ are shown in black;
  they can be regarded as providing decisive evidence against the null hypothesis
  (to use Jeffreys's expression in a slightly different context).
\end{itemize}

Jeffreys \cite[Appendix~B]{Jeffreys:1961} gives a crude but convenient rule of thumb
for comparing Bayes factors and p-values.
E-values are closely related to Bayes factors,
and in fact e-values are Bayes factors (and vice versa) in the case of simple null hypotheses;
see, e.g., \cite[Sect.~1.1.3]{Grunwald/etal:arXiv1906} for a detailed discussion.
According to Jeffreys's rule of thumb as applied to e-values and p-values,
a p-value of $5\%$ corresponds to an e-value of $10^{1/2}\approx3.16$,
and a p-value of $1\%$ corresponds to an e-value of $10$.
This and similar informal correspondences suggested by other authors (such as I.~J.~Good)
are described in detail in \cite[Sect.~3]{Vovk/Wang:2023}.
Therefore, according to this interpretation of Jeffreys's rule,
the discovery matrices and discovery p-matrices
using our colour schemes are somewhat comparable;
e.g., the light red cells correspond
to similar amounts of evidence against the null hypothesis.

Comparing the left panel of Figure~\ref{fig:ie_vs_p}
to the right panel of Figure~\ref{fig:U1_vs_U2},
we can see that the results produced using e-values
are typically much better leading to tighter confidence regions.
Remarkably, even after the crude e-to-p calibration $e\mapsto 1/e$
our method produces p-values that look better than the p-values
produced by the GWGS procedure:
see the right panel of Figure~\ref{fig:ie_vs_p}.

A simple informal explanation of the superior performance of our method
is that when combining independent e-values evidence may multiply.
For instance, the product of two ``non-substantial'' independent e-values (e.g.,  around $3$) leads 
to a ``substantial'' e-value (e.g., around $9$).
In the GWGS procedure, the  standard method of combining independent p-values is that of Simes,
which will not produce anything smaller than the smallest input p-value.
In contrast, our $U_2$ function (as well as $U_n$ for $n\ge 3$)
is able to produce larger output e-values than the largest individual input.
This observation suggests that other versions of the GWGS procedure,
for instance combining a U-statistic with Fisher's method \cite{Fisher:1932},
might be more powerful than the standard method in certain applications.

\section{An empirical study with independent e-values}
\label{sec:empirical}

In this section we will use the \texttt{prostate} dataset first described
by Singh et al.\ \cite{Singh/etal:2002};
it is also analyzed in \cite[Chap.~2]{Efron:2010}
and then in \cite[Chap.~15]{Efron/Hastie:2016}.
This dataset represents a $6033\times102$ matrix
whose rows correspond to 6033 genes and columns correspond to 102 men.
The first 50 men are healthy controls and the remaining 52 are patients with prostate cancer.
Each entry $x_{k,j}$ of the matrix represents the activity of the $k$th gene in the $j$th man.

For each gene we are interested in whether its activity is different
in the patients and the healthy controls.
So we have 6033 null hypotheses of no difference to test.
Efron \cite[Chap.~2]{Efron:2010} makes the assumption of independence
of his test statistics for testing those null hypotheses
(but he also analyses shortfalls of this assumption in \cite[Sect.~2.5]{Efron:2010}),
along with several other substantial assumptions,
such as the Gaussian distribution of the genes' activities.
In this paper, however, we will avoid making any other assumptions apart from independence.

Following \cite[Sect.~8]{Vovk/Wang:2023}
(which, however, considered a different dataset),
we compute the base e-values as
\begin{equation}\label{eq:permutation-e}
  e_k
  :=
  \frac{T_k}
  {
    \frac{1}{B + 1}
    \left(
      \sum_{b=1}^B
      T_k^{(b)}
      +
      T_k
    \right)
  },
  \quad
  k\in\{1,\dots,6033\},
\end{equation}
where $T_k$ is the \emph{nonconformity score} computed
as described in the next paragraph from the $k$th row of the data matrix,
$T_k^{(b)}$ is the nonconformity score computed from the same row with randomly permuted labels
(independently for different $b$),
and $B$ is the number of permutations.

We define the nonconformity score as $T_k:=\left|t_k\right|^d$,
where $d>0$ is a parameter of the algorithm,
\begin{equation*}
  t_k
  :=
  \frac{\bar x_{k,1} - \bar x_{k,0}}{s_k}
\end{equation*}
is the two-sample t-statistic for the $k$th gene,
$\bar x_{k,1}$ is the average entry in the $k$th row for the patients,
$\bar x_{k,0}$ is the average entry in the $k$th row for the healthy controls,
\[
  s^2_k
  :=
  \sum_j
  \left(
    x_{k,j} - \bar x_{k,y_j}
  \right)^2
\]
is the sample variance of row $k$
(ignoring a constant factor, which cancels out when computing~\eqref{eq:permutation-e}),
and $y_j$ is the label (1 for the patients and 0 for the healthy controls).

\begin{figure}
  \begin{center}
    \includegraphics[width=0.49\textwidth]{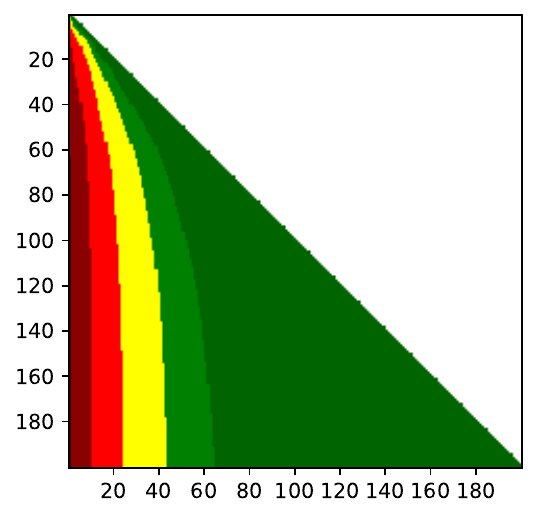}
    \includegraphics[width=0.49\textwidth]{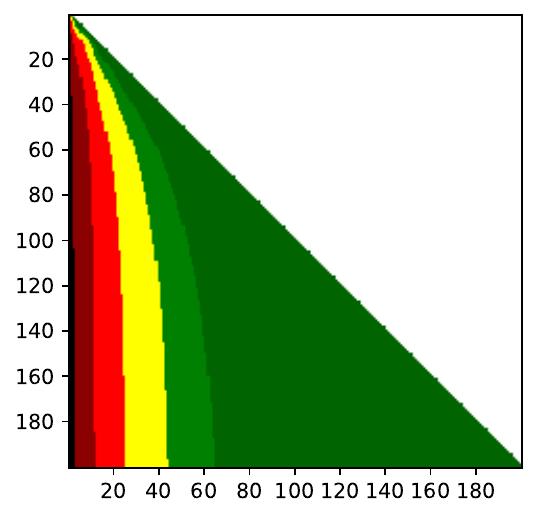}
  \end{center}
  \caption{Left panel:
    the top-left $200\times200$ corner of the $U_2$ discovery matrix
    for the \texttt{prostate} dataset for $B:=10000$,
    using Jeffreys's thresholds.
    Right panel: its simplified version (based on \eqref{eq:permutation-e-simple})
    that is only approximately valid.}
  \label{fig:U2}
\end{figure}

\begin{figure}
  \begin{center}
    \includegraphics[width=0.49\textwidth]{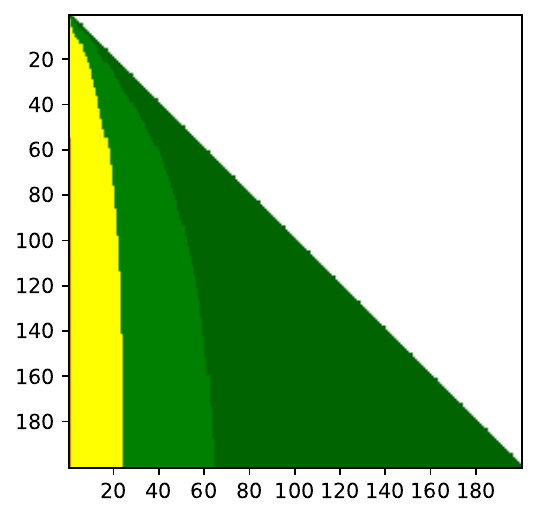}
    \includegraphics[width=0.49\textwidth]{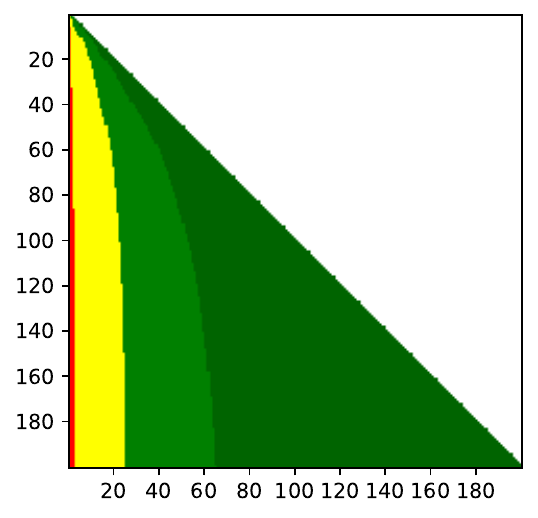}
  \end{center}
  \caption{The analogue of Figure~\ref{fig:U2} for the $U_1$ discovery matrix.}
  \label{fig:U1}
\end{figure}

The left panel of Figure~\ref{fig:U2} gives the top-left $200\times200$ corner of the discovery matrix
for $d:=10$ and $B:=10000$.
The result is much better than for the $U_1$ discovery matrix:
see the left panel of Figure~\ref{fig:U1}.
The price to pay for $U_2$ giving tighter confidence regions than $U_1$ is,
of course,
the reliance of the former on the independence assumption for the base e-variables.

In the right panels of Figures~\ref{fig:U2} and~\ref{fig:U1}
we give analogous plots but with the proper e-values \eqref{eq:permutation-e}
replaced by their commonly used simplified versions
\begin{equation}\label{eq:permutation-e-simple}
  e_k
  :=
  \frac{T_k}
  {
    \frac{1}{B}
    \sum_{b=1}^B
    T_k^{(b)}
  },
  \quad
  k\in\{1,\dots,6033\}.
\end{equation}
We will call~\eqref{eq:permutation-e-simple} \emph{simplified e-values},
but in fact they are not bona fide e-values.
The similarity between the definitions~\eqref{eq:permutation-e}
and~\eqref{eq:permutation-e-simple} suggests
that simplified e-values are approximately valid,
but sometimes their lack of validity is visible.
The closeness of the left and right panels in Figures~\ref{fig:U2} and~\ref{fig:U1}
suggests that the number of iterations $B=10000$ is sufficiently large.

\subsection*{Comparison with methods based on p-values}

\begin{figure}
  \begin{center}
    \includegraphics[width=0.49\textwidth]{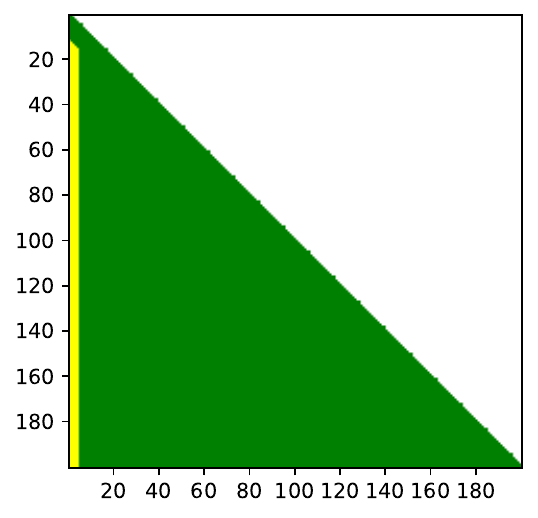}
    \includegraphics[width=0.49\textwidth]{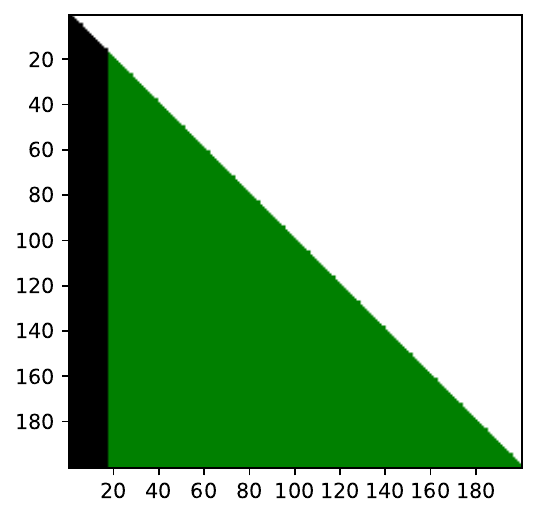}
  \end{center}
  \caption{The analogue of Figure~\ref{fig:U2} for the discovery p-matrix
    using the GWGS method with independent p-values.}
  \label{fig:GWGS}
\end{figure}

Finally, Figure~\ref{fig:GWGS} is the analogue of Figures~\ref{fig:U2} and~\ref{fig:U1}
for the GWGS method \eqref{eq:DM-p}
applied to the p-values
\begin{equation}\label{eq:permutation-p}
  p_k
  :=
  \frac
  {
    \left|\left\{
      b\in\{1,\dots,B\}
      \mid
      T_k^{(b)}
      \ge
      T_k
    \right\}\right|
    +
    1
  }{B+1}
\end{equation}
(cf.\ \cite[(2.20)]{Vovk/etal:2022book})
and their simplified versions
\begin{equation}\label{eq:permutation-p-simple}
  p_k
  :=
  \frac
  {
    \left|\left\{
      b\in\{1,\dots,B\}
      \mid
      T_k^{(b)}
      \ge
      T_k
    \right\}\right|
  }{B}
\end{equation}
for $k\in\{1,\dots,6033\}$.
In this case, there is no dependence on $d$ as the p-values only depend on the ranks of $T_k$.
The valid p-values~\eqref{eq:permutation-p}
used in the left panel of Figure~\ref{fig:GWGS} give a very poor result,
and the comparison with the right-hand panel
shows that the number of iterations $B=10000$ is far too small
when using p-values:
the putative p-values~\eqref{eq:permutation-p-simple}
in the right-hand panel are very far from being valid.

\section{Towards efficient algorithms for $U_2$ and other~$U_n$}
\label{sec:efficient}

\begin{algorithm}[bt]
  \caption{One row of the  discovery matrix $\DM=\DM^{U_2}$ in time $O(K^2)$}
  \label{alg:DM-eff}
  \begin{algorithmic}[1]
    \Require
      Decreasing sequence of e-values $e_1\ge\dots\ge e_K$.
    \Require
      Row number $r\in\{1,\dots,K\}$ of the discovery matrix.
    \State $s_{K+1} := 0$\label{ln:s-start}
    \For{$k=K,\dots,1$}
      \State $s_k := s_{k+1}+e_k$\label{ln:s-end}
    \EndFor
    \For{$j=r-1,\dots,0$}
      \State $V := V_2(\{j+1,\dots,r\})$\label{ln:V}
      \State $e := \frac{2V}{(r-j)(r-j-1)}$ (unless $j=r-1$)\label{ln:e}
      \For{$k=K,\dots,r+1$}\label{ln:it-start}
        \State $V' := V_2(\{j+1,\dots,r\}\cup\{k,\dots,K\})$\label{ln:V-prime}
        \State $e' := \frac{2V'}{(r-j+K-k+1)(r-j+K-k)}$\label{ln:e-prime}
        \If{$e' < e$}
          \State $e := e'$
        \EndIf
      \EndFor
      \State $\DM_{r,j} := e$
    \EndFor
  \end{algorithmic}
\end{algorithm}

In this appendix we will see that each row of the discovery matrix based on $U_2$
can be computed in time $O(K^2)$,
and so the computation of the full discovery matrix takes time $O(K^3)$.
This is not as good as the $O(K^2)$ algorithm for $U_1$
given as Algorithm~4 in \cite{Vovk/Wang:2023},
and the existence of such an algorithm for $U_2$ remains an open problem.

Algorithm~\ref{alg:DM-eff} starts (in lines \ref{ln:s-start}--\ref{ln:s-end})
from defining an array
\[
  s_k
  :=
  e_k+\dots+s_K,
  \quad
  k\in\{1,\dots,K+1\},
\]
of partial sums, which is not used in the algorithm explicitly,
but we will explain how it enables
an efficient update of the variables $V$ and $V'$
in lines~\ref{ln:V} and~\ref{ln:V-prime}.
In the algorithm we use the notation
\[
  V_2(I)
  :=
  \sum_{i,i'\in I: i<i'}
  e_i e_{i'},
  \quad
  \emptyset\subset I\subseteq\{1,\dots,K\};
\]
this expression is the key component of the U-statistic computed from $e_i$, $i\in I$.
When $\left|I\right|=1$,
we will treat the value of $V_2(I)$ as undefined.

The first value of $V$ in line~\ref{ln:V}, $V_2(\{r\})$,
is undefined, and we then set $e:=e_{r}$
instead of the formula given in line~\ref{ln:e}.
The next value of $V$,
\[
  V
  :=
  V_2(\{r-1,r\}
  =
  e_{r-1} e_{r},
\]
is computed from scratch in time $O(1)$,
and the following values are computed in time $O(1)$ using the previous value:
$V = V_2(\{j+1,\dots,r\})$ is computed from the previous value as
\[
  V
  =
  V
  +
  e_{j+1}
  \left(
    e_{j+2}+\dots+e_{r}
  \right)
  =
  V
  +
  e_{j+1}
  \left(
    s_{j+2} - s_{r+1}
  \right).
\]
The very first value of $V'$ computed in line~\ref{ln:V-prime}
(for $j=r-1$ and $k=K$)
can be found in time $O(1)$ from scratch,
\[
  V'
  =
  V_2(\{r,K\})
  =
  e_r e_K.
\]
After that the first value of $V'$ in each execution of the loop
starting in line~\ref{ln:it-start}
(i.e., $V'$ for $j<r-1$ and $k=K$)
can be found in time $O(1)$ from the current value of $V$ using
\[
  V'
  :=
  V
  +
  e_K
  \left(
    s_{j+1} - s_{r+1}
  \right).
\]
For the following iterations of this loop
the value of $V'$ can be updated in time $O(1)$ using
\begin{align*}
  V'
  &:=
  V'
  +
  e_{k}
  \left(
    e_{j+1}+\dots+e_{r}+e_{k+1}+\dots+e_K
  \right)\\
  &=
  V'
  +
  e_{k}
  \left(
    s_{j+1} - s_{r+1} + s_{k+1}
  \right).
\end{align*}

\subsection*{Modifications for $U_n$, $n>2$}

It is easy (but tiresome) to modify Algorithm~\ref{alg:DM-eff}
so that it computes a row of $\DM^{U_n}$ for a fixed $n>2$ in time $O(K^2)$.
Let us consider, for simplicity, the case $n=3$.

Both entries of $V_2$ in Algorithm~\ref{alg:DM-eff} should be changed to $V_3$,
where
\[
  V_3(I)
  :=
  \sum_{i,i',i''\in I: i<i'<i''}
  e_i e_{i'} e_{i''},
\]
line \ref{ln:e} should be changed to
\[
  e
  :=
  \frac{6V}{(r-j)(r-j-1)(r-j-2)}
\]
(unless $j=r-1$ or $j=r-2$),
and line \ref{ln:e-prime} should be changed to
\[
  e'
  :=
  \frac{6V'}{(r-j+K-k+1)(r-j+K-k)(r-j+K-k-1)}
\]
(unless $j=r-1$ and $k=K$).
Lines \ref{ln:s-start}--\ref{ln:s-end} for computing the array
\[
  s_k
  :=
  V_1(\{k,\dots,K\}),
\]
where
\[
  V_1(I)
  :=
  \sum_{i\in I}
  e_i,
\]
should be complemented by computing the array
\[
  t_k
  :=
  V_2(\{k,\dots,K\}),
  \quad
  k=K,\dots,1.
\]
The array $t$ can be computed in time $O(K)$ starting from $t_K:=0$
and setting, for $k=K-1,\dots,1$,
\[
  t_k
  :=
  t_{k+1}
  +
  e_k s_{k+1}.
\]
In line \ref{ln:V}, we can compute $V$ in time $O(1)$ given its previous value using the identity
(true unless $j$ is very close to $r$)
\begin{equation*}
  V_3(\{j+1,\dots,r\})
  =
  V_3(\{j+2,\dots,r\})
  +
  e_{j+1}
  V_2(\{j+2,\dots,r\})
\end{equation*}
and the identity
\begin{equation}\label{eq:ab}
  V_2(\{j+2,\dots,r\})
  =
  t_{j+2} - t_{r+1} - s_{r+1} (s_{j+2} - s_{r+1}).
\end{equation}
Finally, in line \ref{ln:V-prime} we can compute $V'$ in time $O(1)$
given its previous value using the identity
(true unless $j$ is very close to $r$ or $k$ is very close to $K$)
\begin{align*}
  &V_3(\{j+1,\dots,r\}\cup\{k,\dots,K\})\\
  &=
  V_3(\{j+1,\dots,r\}\cup\{k+1,\dots,K\})\\
  &\quad+
  e_{k}
  V_2(\{j+1,\dots,r\}\cup\{k+1,\dots,K\})\\
  &=
  V_3(\{j+1,\dots,r\}\cup\{k+1,\dots,K\})\\
  &\quad+
  e_{k}
  \left(
    V_2(\{j+1,\dots,r\})
    +
    t_{k+1}
    +
    (s_{j+1}-s_{r+1})s_{k+1}
  \right)
\end{align*}
and the identity \eqref{eq:ab}
(with $j+1$ in place of $j+2$).
The simple but numerous special cases (signalled by ``unless'') should be considered separately.

\section{When can we expect the $U_2$ merging function to be effective?}
\label{sec:attempt}

This section is an attempt of a theoretical explanation
of the efficiency of the $U_2$ merging function
in our simulation and empirical studies.
We start from an alternative representation of $U_2$,
which will shed some light on its expected performance.

Let $\mathbf{e}:=(e_1,\dots,e_K)\in[0,\infty)^K$,
$M_1=U_1(\mathbf{e})$ be the arithmetic mean of $e_1,\dots,e_K$,
\[
  M_2
  :=
  \sqrt{\frac{e_1^2+\dots+e_K^2}{K}}
\]
be the quadratic mean of $e_1,\dots,e_K$,
and
\[
  \var(\mathbf{e})
  :=
  \frac1K
  \sum_{k=1}^K
  (e_k - M_1)^2
  =
  M_2^2
  -
  M_1^2
\]
be the sample variance of $e_1,\dots,e_K$.

\begin{lemma}
  For any $\mathbf{e}$,
  \begin{equation}\label{eq:identity}
    U_2(\mathbf{e})
    = 
    M_1^2 - \frac{1}{K-1} \var(\mathbf{e}).
  \end{equation}
\end{lemma}

\begin{proof}
  By definition,
  \begin{align*}
    U_2(\mathbf{e})
    &=
    \frac{1}{K(K-1)/2}
    \sum_{i<j}
    e_i e_j
    =
    \frac{1}{K(K-1)}
    \left(
      \left(
        \sum_i e_i
      \right)^2
      -
      \sum_i e_i^2
    \right)\\
    &=
    \frac{K}{K-1} M_1^2
    -
    \frac{1}{K-1} M_2^2
    = 
    M_1^2 - \frac{1}{K-1} \var(\mathbf{e}).
    \qedhere
  \end{align*}
\end{proof}

\begin{corollary}\label{cor:rvar}
  For any $\mathbf{e}$,
  \[
    \var(\mathbf{e})
    \le
    (K-1)
    M_1^2.
  \]
  For some $\mathbf{e}\ne0$ the inequality holds as equality.
\end{corollary}

\begin{proof}
  The first statement follows from $U_2(\mathbf{e})\ge0$,
  and an example for the second one is $\mathbf{e}:=(K,0,\dots,0)$.
\end{proof}

According to Corollary~\ref{cor:rvar},
\[
  \rvar(\mathbf{e})
  :=
  \frac
  {\var(\mathbf{e})}
  {(K-1)M_1^2},
\]
which we will call the \emph{relative (sample) variance} of $\mathbf{e}$,
is a dimensionless quantity in the interval $[0,1]$.
When $\mathbf{e}=0$, we set $\rvar(\mathbf{e}):=0$.
The relative variance is zero if and only if all $e_i$ coincide,
and it is 1 if and only if all $e_i$ but one are zero.

Using the notion of relative variance,
we can rewrite \eqref{eq:identity} as
\begin{equation}\label{eq:identity-2}
  U_2(\mathbf{e})
  = 
  M_1^2
  (1 - \rvar(\mathbf{e})).
\end{equation}
We can see that the method of this paper based on $U_2$ has a potential
for improving on the $U_1$ method of \cite{Vovk/Wang:2023},
but the best it can achieve is squaring the entries of the discovery matrix.
An entry of a discovery matrix based on $U_1$
is squared when we replace $U_1$ by $U_2$
if the multiset of e-values on which the min in \eqref{eq:algorithm} is attained
(for $F:=U_1$)
consists of a single value.
Otherwise the discovery matrix based on $U_2$ suffers as the e-values become more diverse.

The identity~\eqref{eq:identity-2} shows that $U_2>M_1$ (i.e., $U_2$ improves on $U_1=M_1$)
if and only if
\begin{equation*}
  \rvar(\mathbf{e})
  <
  1 - \frac{1}{M_1}.
\end{equation*}
For example, if the result of applying $U_1$ is borderline strong evidence against the null hypothesis,
$U_1=10$,
it is improved by $U_2$ if and only if $\rvar(\mathbf{e})<0.9$.
For comparison:
\begin{itemize}
\item
  in our simulation studies,
  the relative variance of the whole set of 200 likelihood ratios \eqref{eq:E}
  (those used in Figure~\ref{fig:U1_vs_U2})
  is approximately $0.24$,
  and the relative variance of the 20 largest of them is approximately $0.23$;
\item
  in our empirical studies,
  the relative variance of the whole set of 6033 Monte Carlo e-values~\eqref{eq:permutation-e}
  (those used in Figures~\ref{fig:U2} and~\ref{fig:U1})
  is approximately $0.035$;
  the relative variance of the 200 largest among those 6033 e-values
  is approximately $0.031$,
  while the relative variance of the 20 largest
  is approximately $0.028$.
\end{itemize}

\section{Conclusion}
\label{sec:conclusion}

This paper has given examples of se-merging functions,
namely $U_n$,
which can be successfully applied,
in simulation and empirical studies,
to multiple hypothesis testing with independent base e-values.
An interesting question is whether admissible ie-merging functions
different from $U_n$ and their convex combinations
can also be useful for this purpose.
Such functions definitely exist;
e.g., in \cite[Remark~4.3]{Vovk/Wang:2021}
we show that
\begin{equation}\label{eq:f}
  f(e_1,e_2)
  :=
  \frac12
  \left(
    \frac{e_1}{1 + e_1}
    +
    \frac{e_2}{1 + e_2}
  \right)
  \left(
    1 + e_1 e_2
  \right)
\end{equation}
is an admissible ie-merging function.
To check that $f(e_1,e_2)$ does not have the form $a e_1 + b e_2 + c e_1 e_2$,
it suffices to set $e_2:=0$.
We can even show that \eqref{eq:f} is not an se-merging function:
see \cite[Example~2]{Vovk/Wang:2024}.

In our computational experiments in Sections~\ref{sec:simulation}
and~\ref{sec:empirical} we only study the case of independent base e-values,
while in Sect.~\ref{sec:discovery} we study the more general case
of sequential e-values.
The reason is that, while our main procedure works perfectly well
in the sequential case from the mathematical point of view,
it appears to be natural from the practical point of view
only in the independent case.
In the sequential case,
as described in Sect.~\ref{sec:discovery},
it is natural to allow testing the same hypothesis more than once.
Another interesting direction of further research
is to design such more general sequential procedures.

In Sect.~\ref{sec:efficient} we described an $O(K^2)$ algorithm
implementing our procedure for computing one row of discovery matrices based on $U_n$.
It is an open problem to find $O(K)$ procedures for computing one row of discovery matrices
based on $U_n$ with $n>1$
(for $n=1$ it was done in \cite[Proposition 4.1]{Vovk/Wang:2021}),
or to prove that such procedures do not exist.

\subsection*{Acknowledgments}
\addcontentsline{toc}{section}{Acknowledgments}

We are grateful to Yuri Gurevich for useful discussions.
We also thank the editors and two anonymous referees for helpful comments
about the journal version of this paper.
In our simulation and empirical studies we used Python and \textsf{R},
including the package \texttt{hommel} \cite{Goeman/etal:2019R}.

V.~Vovk's research has been partially supported by Astra Zeneca, Stena Line, and Mitie.
R.~Wang is supported by  the Natural Sciences and Engineering Research Council of Canada 
(RGPIN-2018-03823, RGPAS-2018-522590).

\end{document}